\documentclass[journal]{IEEEtran}
\usepackage[dvips]{color}
\usepackage{epsf}
\usepackage{times}
\usepackage{epsfig}
\usepackage{graphicx}
\usepackage{epstopdf}
\usepackage{pstricks}
\usepackage{amsmath}
\usepackage{amssymb}
\usepackage{amsxtra}
\usepackage{here}
\usepackage{rawfonts}
\usepackage{times}
\usepackage{url}
\usepackage{cite}
\usepackage{caption}
\usepackage{subcaption}
\usepackage{algorithm}
\usepackage{algorithmic}
\usepackage{cleveref}
\usepackage{afterpage}
\usepackage{amsthm}

\usepackage{breakurl}

\newtheorem{theorem}{Theorem}

\ifCLASSINFOpdf
\else
\fi
\begin{document}
\title{Authentication of Everything in the Internet of Things: Learning and Environmental Effects}

\author{\IEEEauthorblockN{Yaman Sharaf Dabbagh and Walid Saad\\\thanks{This research was supported by the U.S. National Science Foundation under Grant CNS-1524634."}}
\IEEEauthorblockA{Wireless@VT, Bradley Department of Electrical and Computer Engineering, Virginia Tech, VA, USA\\
Email: \{yamans, walids\}@vt.edu\vspace{-0.5cm}}}

\maketitle

\begin{abstract}
Reaping the benefits of the Internet of things (IoT) system is contingent upon developing IoT-specific security solutions. Conventional security and authentication solutions often fail to meet IoT security requirements due to the computationally limited and portable nature of IoT objects. In this paper, an IoT objects authentication framework is proposed. The framework uses device-specific information, called fingerprints, along with a transfer learning tool to authenticate objects in the IoT. The framework tracks the effect of changes in the physical environment on fingerprints and uses unique IoT environmental effects features to detect both cyber and cyber-physical emulation attacks. The proposed environmental effects estimation framework is proven to improve the detection rate of attackers without increasing the false positives rate. The proposed framework is also shown to be able to detect cyber-physical attackers that are capable of replicating the fingerprints of target objects which conventional methods are unable to detect. A transfer learning approach is proposed to allow the use of objects with different types and features in the environmental effects estimation process to enhance the performance of the framework while capturing practical IoT deployments with diverse object types..
Simulation results using real IoT device data show that the proposed approach can yield a 40\% improvement in cyber emulation attacks detection and is able to detect cyber-physical emulation attacks that conventional methods cannot detect. The results also show that the proposed framework improves the authentication accuracy while the transfer learning approach yields up to 70\% additional performance gains.
\end{abstract}


%% Attacks:
%Data Modification
%Identity Spoofing
%Man-in-the-Middle Attack

% \fi
\IEEEpeerreviewmaketitle

%\vspace{-0.1cm}
\section{Introduction}
\IEEEPARstart{T}{he} Internet of things~(IoT) is a rapidly emerging paradigm in which physical objects integrate with the cyber world via smart sensors, RFID tags, smartphones, and wearable devices~\cite{park2016learning}. This integration allows physical objects to operate over the Internet so as to collect and exchange data that describe the physical world. The wide variety of cyber-enabled objects remotely operating through various types of networks and protocols raises many serious security and privacy concerns~\cite{weber2010internet}. %could add more IoT intro and its growth
Security threats range from physical attacks to attacks on the semantic application layers where information is processed and analyzed.
One key challenge is that most IoT objects operate at low energy levels with minimal computation capabilities, and thus, require simple security solutions~\cite{atzori2010internet}. Therefore, most complex security techniques, such as conventional cryptography, firewalls, and secure protocols cannot be readily implemented in the IoT due to the strict memory and computing requirements of its devices.

%One of these low energy objects is the RFID, which is one of the important building blocks for IoT. Securing RFID techniques are either physical based, such as adding physical trigger or pin on the tag to kill/activate the RFID~\cite{stajanoresurrecting}, or cryptographic based which require a level of computational power, such as the minimalist cryptography for RFID tags proposed by~\cite{juels2005minimalist}. 

% reference that fingerprints can be replicated or hacked: [19] B. Danev, H. Luecken, S. Capkun and K. El Defrawy, "Attacks on physical-layer identification," Proceedings of the third ACM conference on Wireless network security, p. 2010, 89-98. [20] M. Edman and B. Yener, "Active attacks against modulationbased radiometric identification," RPI Department of Computer Science Technical Report, 2009.

Prior research on IoT security has primarily focused on two main tracks: creating lightweight security methods~\cite{wu2017efficient,yang2016towards,ferdowsi2017deep,patel2015non,bojinov2014mobile,radhakrishnan2013passive} and building secure IoT architectures~\cite{raza2017securesense,stergiou2018secure}.
Efficient authentication of objects in IoT systems is a challenge due to the low computing capabilities of IoT objects. The authors in~\cite{wu2017efficient} proposed an authentication scheme that uses a key change method in order to improve the security of IoT objects while using shorter and less complicated security keys. The authors in~\cite{yang2016towards} focused on anonymous entity authentication and proposed a lightweight scheme for IoT systems. Their proposed scheme used a dynamic accumulator for credentials that solves the issue of credentials update which requires computational power from IoT objects. As for building secured IoT architectures, the authors in~\cite{raza2017securesense} proposed a secure communication architecture specifically designed for cloud-connected IoT objects. Their proposed architecture includes an end-to-end secure communication between low power IoT objects and cloud back-ends. The authors in~\cite{stergiou2018secure} proposed another architecture to secure IoT objects that offloads the computations needed for authentication to the cloud in order to reduce the overhead on IoT objects. However, all the previous approaches and architectures impose high computational requirements on IoT objects for handling cryptographic keys and credentials exchange which some basic IoT objects are not able to process. One promising approach to protect wireless devices with minimal to no computational load on IoT objects is by analyzing IoT signals. This approach has two main types, \emph{signal watermarking} and \emph{device fingerprinting}. In signal watermarking, a predefined signal is watermarked into IoT object signals. In~\cite{ferdowsi2017deep}, the authors presented an approach for dynamic watermarking of IoT signals using deep learning. Meanwhile, device fingerprinting is a technique to authenticate devices using unique features extracted from the objects transmitted signals.
Such fingerprinting can be done with minimal computational overhead. The authors in~\cite{patel2015non}, used an object's RF-emissions as fingerprints to authenticate ZigBee devices. As for mobile devices fingerprinting, there exist many other features other than the RF-emissions to fingerprint. The authors in~\cite{bojinov2014mobile} used features such as accelerometer calibration error and microphone distortion as fingerprints to identify mobile devices. Device fingerprinting is not necessarily implemented on features extracted from the device/object. For example in~\cite{radhakrishnan2013passive}, network traffic features, such as packet inter-arrival time, and delays between successive packets are used as features to identify devices/objects in a network of devices.

These existing device fingerprinting techniques face three main limitations. First, fingerprinting features are assumed to be the same across all devices in the system, such as in~\cite{patel2015non,bojinov2014mobile,radhakrishnan2013passive}. This assumption is not practical because an IoT system consists of a wide variety of object types with different features.
%\begin{figure}[!t]
%\centering
% 	   \includegraphics[width=8.5cm]{Figures/Drawing6test.png}
%  	  \caption{\label{fig:system}Device fingerprinting in the IoT system}
%\end{figure}
Second, in techniques such as~\cite{patel2015non}, devices are required to be connected directly to a central sensing node that extracts fingerprinting features. These sensing nodes are assumed to be resistant to attacks and device fingerprinting algorithms are not applied on these sensing nodes. Third, existing device fingerprinting  techniques such as~\cite{xu2016device} assume fingerprints are fixed and do not change over time. However, in an IoT environment, fingerprints of objects change with time due to multiple factors, such as changes in the surrounding environment, aging of objects, and noise, as shown in~\cite{bertoncini2012wavelet}. To the best of our knowledge, these changes in fingerprints have not been exploited as a feature to authenticate objects in IoT systems.
%Another challenge in the IoT is its hierarchal structure, in other words, most objects are not connected directly to the Internet but require having one or more relay objects to be able to connect to the Internet. 
Therefore, a direct implementation of existing fingerprinting techniques, such as~\cite{patel2015non,bojinov2014mobile,radhakrishnan2013passive}, in practical IoT systems is not feasible.

The main contribution of this paper is a novel IoT object authentication framework that can distinguish between signals from legitimate IoT objects and signals from malicious objects. 
The proposed framework exploits the effects of the environment surrounding IoT objects to build a model for the expected environmental effects on each object in the IoT system. This model is used to distinguish remote attackers such as cyber emulation attackers. The environment model is able to detect highly intelligent attackers capable of replicating the exact fingerprints of IoT objects, we refer to this type of attacks as cyber-physical emulation attacks. The proposed framework tracks the changes in fingerprints for all IoT objects and utilizes the similarities in these changes to extract a model for the environment. Using the estimated environmental effects, emulation attackers that replicate IoT objects in a remote location will not be able to replicate all of the changes in the environment hence allowing our approach to effectively detect them. The novel environmental estimation in our framework enhances the authentication of legitimate objects as well which yields to high detection rate for cyber emulation attacks without increasing the false positives of misclassified legitimate objects. Additionally, our framework uses transfer learning to estimate the environment from objects of different types or objects with different feature spaces. The ability to use objects of different feature spaces is important in real IoT systems with wide diversity of objects. 
To our knowledge, this work is the \emph{first to exploit the changes in environmental effects on IoT object fingerprints}. Simulation results using real IoT device data show that our proposed framework enhances the detection rate of cyber emulation attacks. Moreover, the results show that conventional methods of device fingerprinting are unable to detect cyber-physical emulation attackers while our proposed framework were able to detect these type of attacks. The results shows an improvement of 40\% in cyber emulation attacks detection. The transfer learning results shows an improvement of 70\% when our framework uses different types of objects compared to using our framework only on objects with the exact same feature space.

The rest of this paper is organized as follows. Section~\ref{sec2} presents the system model and the proposed framework. Section~\ref{sec:estenv} describes the proposed environmental estimation approach. In Section~\ref{sec:tranlern}, we present the transfer learning approach. Simulation results and evaluation are presented in Section~\ref{sec:results} while conclusions are drawn in Section~\ref{sec:end}.

%%%%%%%%%%%%%%%%%%%%%%%%%%%%%%%%
\vspace{-0.1cm}
\section{System Model and Proposed Framework}\label{sec2}
%\subsection{Adversary Model}
%system model IoT system architecture (assume the base station is trustful and won’t be compromised) heterogeneous, dynamic
%features
%attack model
%IGMM model
%detection scheme

\begin{figure}[!t]
  \begin{center}

%\begin{subfigure}{.48\textwidth}
	\centering
 	   \includegraphics[width=7.7cm]{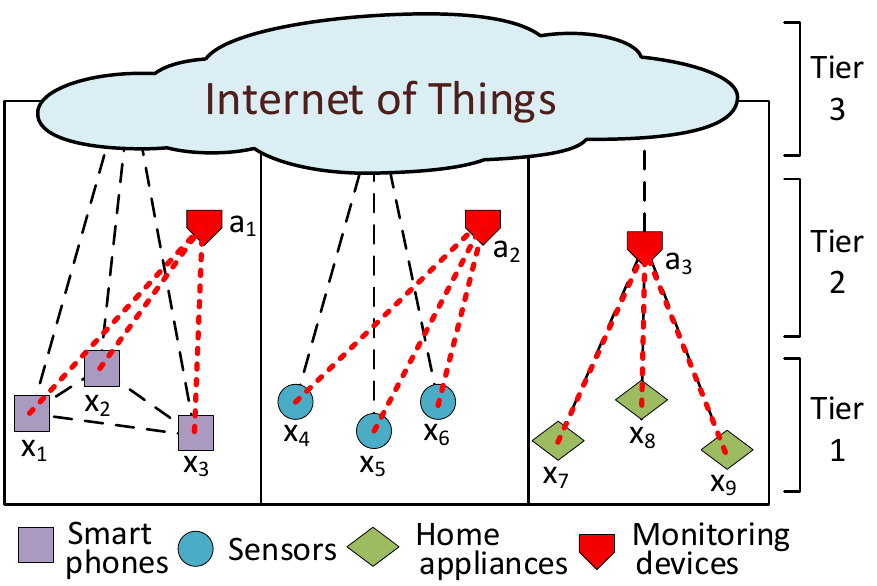}
%  	  \caption{Actual IoT system  \label{fig:type}}
%  \end{subfigure}
%  \begin{subfigure}{.48\textwidth}
%	\centering
% 	   \includegraphics[width=6cm]{Figures/simple_model.pdf}
%  	  \caption{Proposed system model\label{fig:model} }
%	
%  \end{subfigure}
  
  \caption{\small{\label{fig:samplesys} Illustrative example of the considered IoT system consisting of three tiers: objects tier, monitoring devices tier, and cloud tier. The first tier in the example is showing three types of objects: smart phones, sensors, and home appliances, along with three possible setups for the monitoring devices.}}

%	\centering
% 	   \includegraphics[width=8.7cm]{Figures/Model7.pdf}\vspace{-0.2cm}
%  	  \caption{\small{\label{fig:type}Different types of IoT objects}}
  \end{center}
\end{figure}

Consider an IoT system consisting of $N$ heterogeneous objects $x_1,x_2,\dots,x_N$  as shown in Fig.~\ref{fig:samplesys}. These objects can  represent any type of IoT devices such as sensors, smartphones, home appliances, or RFID tags~\cite{dawy2017toward}.
Each object transmits data to a gateway router through a wired or wireless link, such as 802.11 or Zigbee. The gateway aggregates the data transmitted from the IoT objects and forwards it to the cloud where a security service provider (SP) has access to the transmitted data. Security SP is a client on the cloud that generates and enforces network access for connected IoT objects. The SP authenticates the transmitted data from IoT objects based on device-specific information, called \emph{fingerprints}, that uniquely identifies each object in the IoT system. For example, wavelet-based features such as mean, variance, and skewness of electronic codes are possible fingerprints for RFID tags~\cite{bertoncini2012wavelet}. Such wavelet-based features of the electronic codes of RFID are unique for each RFID tag even if the tags are made by the same factory with the same specifications due to hardware impairments during the manufacturing process~\cite{danev2012physical}.%add how information is collected

In this system, an adversary attempts to impersonate a legitimate IoT object to inject tampered information into the IoT system. We consider two classes of adversaries: First, adversaries capable of emulating the software of a legitimate IoT object, which includes emulating security keys, device addresses, and transmitted data type. However this type of adversaries is not able to replicate the hardware features of a legitimate IoT object. Hereinafter, we will refer to this class of attacks as \emph{cyber emulation attack}. The second class of adversaries are capable of emulating both the software and hardware of a legitimate IoT object. This class of attackers is assumed to be highly skilled and is able to replicate the object's software, such as object security keys and object network address, and also clone the legitimate object's device-specific information such as transmission speed, signal strength, processing speed, and operating temperature and humidity. We will refer to this class of attacks as \emph{cyber-physical emulation attack}. These two kinds of attacks compromise the security of the IoT system by replicating the signals and/or fingerprints of legitimate IoT objects; therefore, authenticating the received objects fingerprints based on features that attackers cannot replicate is important to overcome emulation attacks. One of the features that emulation attackers cannot replicate are the environmental changes that pertain the environment that surrounds IoT objects, such as changes in temperature, humidity, wind, physical displacement, or any physical changes affecting IoT objects. These environmental changes are time dependent and are constantly changing which makes it difficult to replicate by a remote emulation attacker. In addition, cyber-physical emulation attacks are difficult and computationally intensive, hence these kind of attacks cannot be updated on the fly to cope with the environmental changes even if the attacker manages to consistently monitor a certain object~\cite{edman2009active}.

To thwart such emulation attacks, we propose a framework whose goal is to determine whether an IoT object is legitimate or not by analyzing the hardware and software features of objects, this can stop cyber attackers that cannot replicate the hardware of IoT objects. Moreover, by analyzing the changes in objects' features over time due to environmental effects, we can thwart cyber-physical attacks that can replicate the hardware and software of IoT objects but cannot replicate the exact environment surrounding the legitimate objects. The proposed framework estimates the changes on the object-specific information with time which are caused by environment effects on IoT objects. Then, the framework compares these estimates with the actual changes in the object-specific information to determine whether the object is physically present in the environment or located remotely by an adversary.

%\section{Proposed Framework}
%\subsection{System Design}
The proposed framework consists of five main components as shown in Fig.~\ref{fig:frame}: \emph{feature extraction}, \emph{fingerprint generation}, \emph{similarity measure}, \emph{environment estimation}, and \emph{transferring knowledge}. In order to estimate the environmental effects on IoT object fingerprints, first we need to extract the features from each IoT object and generate unique fingerprints from the extracted features. In the following subsection, we explain in detail the process of feature extraction and generation, along with the similarity measure used to compare fingerprints.

\begin{figure}[!t]
  \begin{center}

%\begin{subfigure}{.48\textwidth}
	\centering
 	   \includegraphics[width=8.5cm]{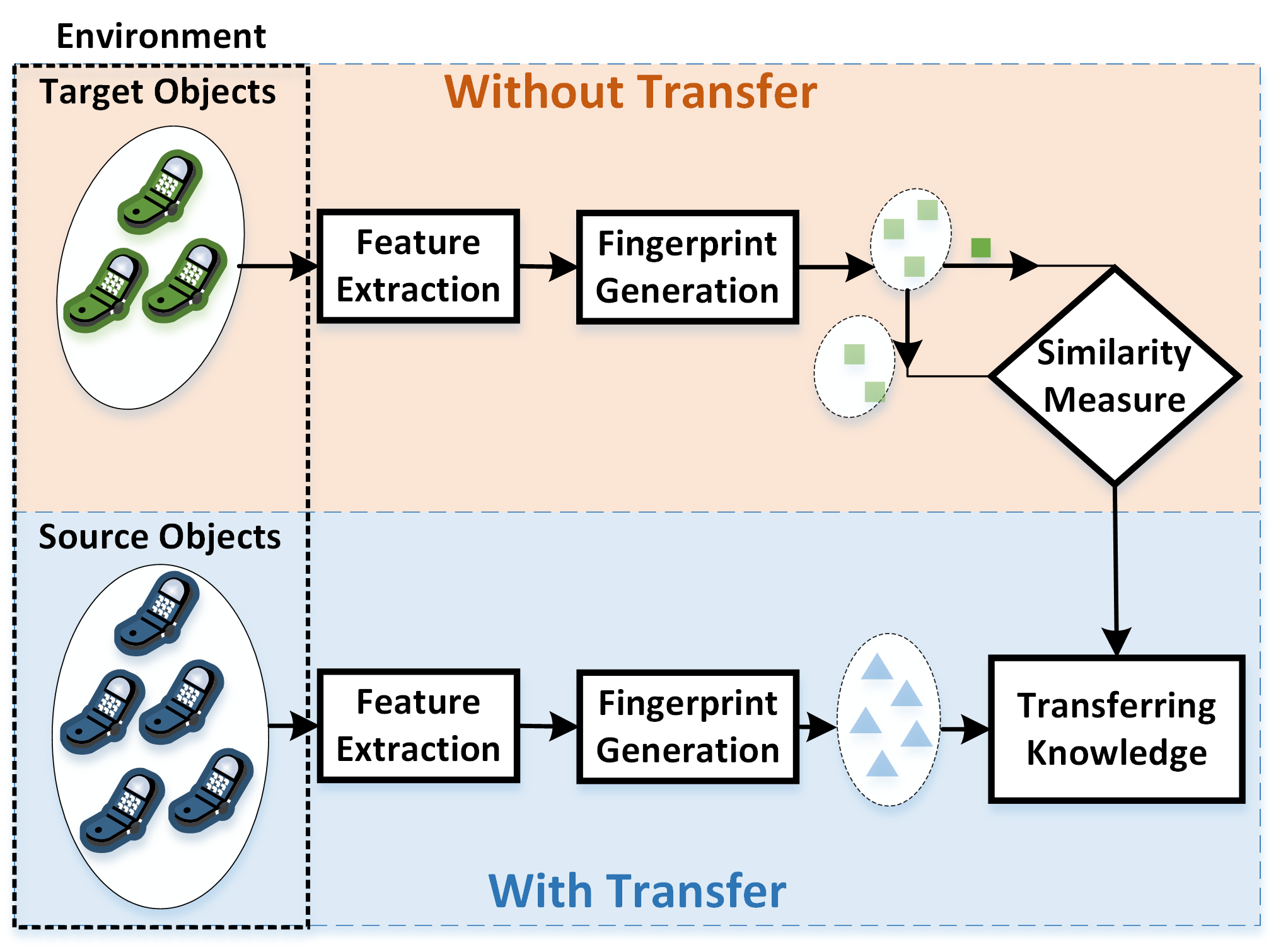}
%  	  \caption{Actual IoT system  \label{fig:type}}
%  \end{subfigure}
%  \begin{subfigure}{.48\textwidth}
%	\centering
% 	   \includegraphics[width=6cm]{Figures/simple_model.pdf}
%  	  \caption{Proposed system model\label{fig:model} }
%	
%  \end{subfigure}
  
  \caption{\small{\label{fig:frame} The main components of the IoT environment estimation framework}}

%	\centering
% 	   \includegraphics[width=8.7cm]{Figures/Model7.pdf}\vspace{-0.2cm}
%  	  \caption{\small{\label{fig:type}Different types of IoT objects}}
  \end{center}
\end{figure}

\vspace{-0.2cm}
\subsection{Feature extraction}IoT object features are collected at three different tiers as shown in Fig.~\ref{fig:samplesys}. The first tier of features is collected at the object level. Each object in the IoT system sends a set of features describing the operational status of the object itself along with all sensor measurements available. Some of the features that can be collected at the object level are CPU load, clock skew, memory usage, and temperature of the object and/or surroundings, among others~\cite{borgia2014internet}. The second tier of feature collection is done at security monitoring objects. These monitoring objects are distributed over the IoT  to gather features about other objects in the IoT system. Gateways are examples of monitoring objects that capture traffic properties of other objects. Some of the features that the monitoring objects can collect are signal strength, signal spectral features, and packet arrival times. The monitoring objects are also considered IoT objects, hence they collect and send first tier features about themselves to the IoT cloud center as well. The third tier of features is collected at the IoT server side by measuring the traffic properties of objects and frequency of received packets. 

\vspace{-0.2cm}
\subsection{Fingerprint generation}The information collected during the feature extraction stage is general and not discriminant for successful object authentication. Hence, we use statistical analysis to select features that are unique for each object. Statistical analysis requires capturing a decent amount of features from objects to form a training dataset. The training dataset is used to find the subset of features that can uniquely identify each object in the IoT system. This subset of discriminant features, called \emph{fingerprint}, forms a vector defined as:
\begin{align}\label{eq:f1}
\boldsymbol{f}^{i,t}_k=[\delta_{k,1},\delta_{k,2},\delta_{k,3},\dots,\delta_{k,m}],
\end{align}
where $\boldsymbol{f}^{i,t}_k$ is the fingerprint vector for object $i$ at time $t$, $\delta_{k,j}$ is the scalar value of a feature $j$ of the fingerprint $\boldsymbol{f}^{i,t}_k$ such as the signal mean feature, and $m$ is the total number of features for object $i$.

\vspace{-0.2cm}
\subsection{Similarity measure}During a time interval $t$, each IoT object $i$ generates $n$ fingerprints $\boldsymbol{f}^{i,t}_1,\boldsymbol{f}^{i,t}_2,\dots,\boldsymbol{f}^{i,t}_n$. Hence, at a time interval $t$, the collected fingerprints of object $i$ form an $n\times m$ matrix $\boldsymbol{F}_{i,t}$ given by:
\begin{align}\label{eq:f2}
\boldsymbol{F}_{i,t} = [\boldsymbol{f}^{i,t}_1,\boldsymbol{f}^{i,t}_2,\dots,\boldsymbol{f}^{i,t}_n]^T,
\end{align}
\begin{align}\label{eq:fmat}
\boldsymbol{F}_{i,t}=\begin{bmatrix}
    \delta_{1,1}       & \delta_{1,2} & \delta_{1,3} & \dots & \delta_{1,m} \\
    \delta_{2,1}       & \delta_{2,2} & \delta_{2,3} & \dots & \delta_{2,m} \\
    \dots &\dots& \dots& \dots& \dots \\
    \delta_{n,1}       & \delta_{n,2} & \delta_{n,3} & \dots & \delta_{n,m}
\end{bmatrix}.
\end{align}

The similarity measure block in Fig.~\ref{fig:frame} at each time interval $t$ measures the distance between the fingerprint matrix $\boldsymbol{F}_{i,t}$ for object $i$ and a reference fingerprint $\overset{*}{\boldsymbol{F}}_{i}$ for the same object $i$. The reference fingerprint $\overset{*}{\boldsymbol{F}}_{i}$ is determined based on previously collected training data. The labeled training data used to determine the reference fingerprints $\overset{*}{\boldsymbol{F}}_{i}$ consists of historical instances of fingerprints for objects in the IoT system. The longer the historical data available for training, the easier it is to spot and remove trends in the data which are due to environmental changes. The reference fingerprints are then chosen from the training dataset that best represent each object in the IoT system.
One of the distance measures to compare two sets of fingerprints is the Bhattacharyya distance measure as follows:

\vspace{-0.3cm}
\begin{align}
\Delta_{i,t} = D_B(\boldsymbol{F}_{i,t},\overset{*}{\boldsymbol{F}}_{i}),\label{eqn:2}
\end{align}
where $\Delta_{i,t}$ is the distance between the collected fingerprint matrix $\boldsymbol{F}_{i,t}$ and the reference fingerprint $\overset{*}{\boldsymbol{F}}_{i}$, $D_B$ is the Bhattacharyya distance measure~(BDM)~\cite{kailath1967divergence}. Other distance measures can be used instead of the BDM to measure $\Delta_{i,t}$. However, the Bhattacharyya distance, unlike other measures such as the KS-test, Hellinger distance, or KL-divergence, can be applied to any type of distribution and can also be applied to both univariate and multivariate distributions~\cite{evren2012some}.
The authentication of any IoT object $i$ is then modeled as a binary hypothesis test:

\begin{equation}\label{eqn:hyp}
\begin{cases}
\quad\mathcal{H}_0:&\Delta_{i,t}\leq \tau,\\
\quad\mathcal{H}_1:&\Delta_{i,t}> \tau,
\end{cases}
\end{equation}
where $\tau$ is a similarity threshold. Therefore, if the difference between $\boldsymbol{F}_{i,t}$ and $\overset{*}{\boldsymbol{F}}_{i}$ is less than $\tau$, then hypothesis $\mathcal{H}_0$ is claimed, and the collected fingerprints are from the legitimate object $i$. On the other hand, if the difference between $\boldsymbol{F}_{i,t}$ and $\overset{*}{\boldsymbol{F}}_{i}$ is larger than $\tau$, a potential attacker is detected and $\mathcal{H}_1$ is applied.
The Bhattacharyya distance used to measure the distance between the distributions $\boldsymbol{F}_{i,t}$ and $\overset{*}{\boldsymbol{F}}_{i}$ can be written as:
\begin{align}
D_B({ \boldsymbol{F}_{i,t},\overset{*}{\boldsymbol{F}}_{i}}) = -\text{ln}\left(\sum\sqrt{\boldsymbol{F}_{i,t},\overset{*}{\boldsymbol{F}}_{i}}\right).
\end{align}

For the special case where both distributions $\boldsymbol{F}_{i,t}$ and $\overset{*}{\boldsymbol{F}}_{i}$ are Gaussian distributions, the Bhattacharyya distance can be written as:
%%%%%%%%%%%%%%%%%%%%%%%%%%%%%%%%
\begin{align}
D_B({ \boldsymbol{F}_{i,t},\overset{*}{\boldsymbol{F}}_{i}})=\frac{1}{8}(\boldsymbol{\mu}_{\scriptstyle \boldsymbol{F}_{i,t}}-&\boldsymbol{\mu}_{\scriptstyle{\overset{*}{\boldsymbol{F}}_{i}}})^T\boldsymbol{\Sigma}^{-1}(\boldsymbol{\mu}_{\scriptstyle \boldsymbol{F}_{i,t}}-\boldsymbol{\mu}_{\scriptstyle{\overset{*}{\boldsymbol{F}}_{i}}})\nonumber\\
&+\frac{1}{2}\ln{\frac{\det\boldsymbol{\Sigma}}{\sqrt{\det\boldsymbol{\Sigma}_{\scriptstyle \boldsymbol{F}_{i,t}}\det\boldsymbol{\Sigma}_{\scriptstyle{\overset{*}{\boldsymbol{F}}_{i}}}}}},
\end{align}
%%%%%%%%%%%%%%%%%%%%%%%%%%%%%%%%%%%
where  $\boldsymbol{\mu}_{\scriptstyle \boldsymbol{F}_{i,t}}, \boldsymbol{\mu}_{\scriptstyle{\overset{*}{\boldsymbol{F}}_{i}}},\boldsymbol{\Sigma}_{\scriptstyle \boldsymbol{F}_{i,t}}$ and $\boldsymbol{\Sigma}_{\scriptstyle{\overset{*}{\boldsymbol{F}}_{i}}}$ are the means and covariances for fingerprints $\boldsymbol{F}_{i,t}$ and $\overset{*}{\boldsymbol{F}}_{i}$ respectively, and $\boldsymbol{\Sigma}=(\boldsymbol{\Sigma}_{\scriptstyle_{\boldsymbol{F}_{i,t}}}+\boldsymbol{\Sigma}_{\scriptstyle_{\overset{*}{\boldsymbol{F}}_{i}}}) / 2$.
Hence, using the distance measure  $D_B$, the distance $\Delta_{i,t}$ between the collected fingerprint  $\boldsymbol{F}_{i,t}$ and the reference fingerprint $\overset{*}{\boldsymbol{F}}_{i}$ determines whether the collected fingerprint $\boldsymbol{F}_{i,t}$ belongs indeed to object $i$ or not. 

While generally it has been assumed that it is impossible to accurately replicate an object's physical fingerprint, an attacker can generate a ``close enough'' fingerprint to the physical object fingerprint. For example, in~\cite{edman2009active} a software-defined radio was used to construct radiometric signatures to impersonate an 802.11b wireless device. Therefore, relying on the similarity threshold $\tau$ alone to determine if fingerprints are legitimate or not is not enough, because it is possible to generate a malicious fingerprint $\boldsymbol{F}_m$ that satisfies $D_B(\boldsymbol{F}_m,\overset{*}{\boldsymbol{F}}_{i})\leq \tau$ which results in considering the malicious fingerprints as legitimate fingerprints. Additionally, reducing the value of the similarity threshold $\tau$ to detect the ``close enough'' fingerprints generated by malicious users can lead to an increase in the number of false positives since various noise sources affects the fingerprint generation and feature extraction processes. Hence, an accurate generation of object fingerprints is required in order to reduce the value of the similarity threshold $\tau$ without compromising the authentication accuracy, where the optimal value for the similarity threshold $\tau$ is the one that achieves the highest true positives and true negatives along with the least false positives and false negatives during the training period. The effect of choosing different values for the similarity threshold $\tau$ will be evaluated in Section~\ref{sec:results}. After generating fingerprints for each IoT objects and defining a similarity measure to compare between fingerprints, we next explain the process used to estimate the environmental effects on objects from the generated fingerprints.

\vspace{-0.2cm}
\section{Environment Estimation}\label{sec:estenv}
During the feature extraction process, object fingerprint features are affected by multiple effects, such as signal interference, electromagnetic radiation, network traffic, and noise. While precisely estimating noise sources is usually not possible due to their random nature and undefined spectrum, having a shared noise source between more than one object can help estimating the amount of shared noise. For example, objects in close proximity to each other are affected by the same environmental changes which affects the fingerprint features for these objects. For example, the authors in~\cite{bertoncini2012wavelet} showed that water submersion and physical crumpling effects on RFID tags impacted the RFID signal of neighboring tags. The environmental effects on a given object's fingerprint features can be represented as a transformation of fingerprints $T$, and, hence, the generated fingerprint matrix $\boldsymbol{F}_{i,t}$ for an object $i$ in~(\ref{eq:fmat}) is:
\begin{align}\label{eq:1}
\hat{\boldsymbol{F}}_{i,t}=T^{-1}(\boldsymbol{F}_{i,t}),
\end{align}
where $\hat{\boldsymbol{F}}_{i,t}$ is the estimated fingerprint matrix without the effect of environment and $T^{-1}$ is the inverse of the transformation $T$. The estimated fingerprint matrix $\hat{\boldsymbol{F}}_{i,t}$ is a more accurate estimation of the fingerprint matrix $\boldsymbol{F}_{i,t}$ because it excludes the effects of environment and is closer to the reference fingerprint matrix $\overset{*}{\boldsymbol{F}}_{i}$. 
%To find the value of the estimated fingerprint matrix $\hat{\boldsymbol{F}}_{i,t}$:
%\begin{align}\label{eq:1}
%\hat{\boldsymbol{F}}_{i,t}=T^{-1}(\boldsymbol{F}_{i,t}),
%\end{align}
%where $T^{-1}$ is the inverse of the transformation $T$. 
The environmental effects on objects are caused by many physical factors such as ambient temperature, humidity, wind, and many other factors surrounding the objects. Most of these effects are non-linear which makes the overall environmental effect on objects non-linear and impossible to predict or estimate.  However, non-linear high-order estimation of the environmental effects suffers from high variance which results in an overfitting transformation. Meanwhile, linear estimations have low variance and high bias which makes the estimation more generalized. here, our goal is to find an estimation of the environmental effects that can impact multiple objects. Hence, a more generalized estimation, such as a linear model, is more likely to capture the common effect that impacts multiple objects at the same time based on the bias-variance tradeoff. Therefore, we consider a linear model to estimate the environmental effects on multiple objects in our framework. This assumption reduces the risk of overfitting to a single object by considering  low-order environmental effects. Consequently, the transformation $T$ can be defined as a rotation and a translation:
\begin{align}\label{eqn:fhat10}
\boldsymbol{F}_{i,t}=\boldsymbol{R}_{i,t}\, \hat{\boldsymbol{F}}_{i,t} + \boldsymbol{l}_{i,t}+\boldsymbol{w},
\end{align}
where $\boldsymbol{R}_{i,t}$and  $\boldsymbol{l}_{i,t}$ are the rotation and translation matrices for object $i$ at time $t$ which transform the fingerprint matrix $\hat{\boldsymbol{F}}_{i,t}$ to align it with the collected fingerprint matrix $\boldsymbol{F}_{i,t}$, and $\boldsymbol{w}$ is a noise which can be a combination of thermal, static, solar, or other noise sources that are purely random and cannot be estimated. Since the noise $\boldsymbol{w}$ cannot be estimated we will ignore the noise during the estimation of the environmental effects in this section.

To find the transforms $\boldsymbol{R}_{i,t}$ and $\boldsymbol{l}_{i,t}$ for object $i$, we use the fingerprint matrices of the objects surrounding the object $i$ to estimate the environment effect on  object $i$. To determine which objects are considered to be surrounding a given object $i$, we use a network graph as shown in Fig.~\ref{fig:net}, where nodes represent IoT objects and edges $\beta_k$ represent the environment similarity between objects. The environment similarity $\beta_k$ between a pair of objects represents the amount of similarity between the environmental effects for each object. The environment similarity $\beta_k$ defers from one pair of objects to the other based on the physical location of objects, and the type of objects, where objects of the same type and located in a close proximity to each other have higher environment similarity between them compared to objects located far apart or objects of different types. The network graph is generated during the training stage by measuring the similarity between environmental effects on objects. For example, to estimate the environmental effects on object $i$ shown in Fig.~\ref{fig:net}, we extract from the full IoT system graph shown in Fig.~\ref{fig:net}(a) a subgraph $\mathcal{G}$ of all of the objects having direct environmental similarity with object $i$ shown in Fig.~\ref{fig:net}. We refer to these objects as neighbors of object $i$. Hence, from ~(\ref{eqn:fhat10}) the fingerprints for each neighboring object can be given by:
\begin{align}
\boldsymbol{F}_{k,t}=\boldsymbol{R}_{k,t}\, \hat{\boldsymbol{F}}_{k,t} + \boldsymbol{l}_{k,t};\,\,\,\,\,\,\,\,\, \forall k \in \mathcal{G}\setminus \{i\},
\end{align}
where $\mathcal{G}\setminus \{i\}$ is the set of objects in the subgraph shown in Fig.~\ref{fig:net}(b) except the object $i$, i.e., the set of all the neighbors of object $i$. 
\begin{figure}[!t]
  \begin{center}
	\centering
 	   \includegraphics[width=6.5cm]{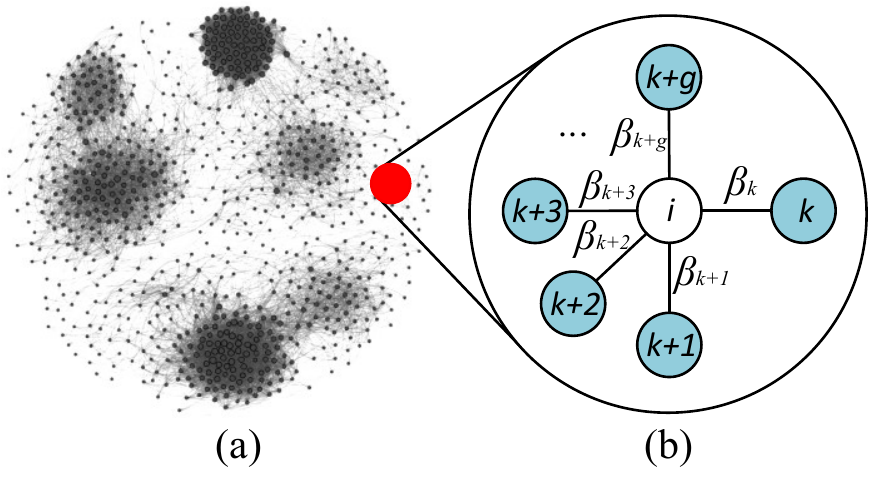}
  	  \caption{\label{fig:net} (a) The network graph of the whole IoT system (b) The graph of the objects having direct environment similarity with object $i$, where $\beta_k$ is the amount of similarity}
  \end{center}
\end{figure}

In order to find the environment effects on object $i$ we assume that the fingerprints received from the neighboring objects are mostly from legitimate objects. This assumption is justified in an IoT setting due to the following reasons. First, the number of objects in the IoT is large which makes it hard for an adversary to attack a significant number of objects at the same time. Second, the attacker does not know which objects are considered by the authentication mechanism as neighbors to a specific object in order to focus its attack resources on the neighboring objects. The reason attackers cannot determine the neighbors of a certain object is due to the nature of the process used to select the neighbors for each object in the IoT system. The neighbors of a target object are determined by analyzing the historical data of all the objects within a certain distance from the target object and finding all the objects that share similar behaviors regardless of their type or exact distance from the target object. This process to choose neighbors is impossible for the attacker to replicate since it requires having access to all the historical data from all the objects in the IoT and also requires knowing the parameters used by the framework behavior analysis. Third, since the network graph is highly connected, even if the attacker manages to attack all of the neighboring objects of a target object in order to cause a wrong environment estimation at the target object, the attack on the neighboring objects will be detected during the estimation process of other target objects in the network. For example, if a target object has five neighboring objects all controlled by a single attacker which allows the attacker to trigger wrong environment estimations at the target object. However, each one of the five neighboring objects will also be participating as a neighbor in the estimation process of other target objects. These other target objects have legitimate objects as neighbors besides the attacked objects which will help to reveal those attacked objects.
Therefore, based on the previous discussions, the fingerprints of each neighboring object $k$ can be given by:
\begin{align}
\boldsymbol{F}_{k,t}=\boldsymbol{R}_{k,t}\, \overset{*}{\boldsymbol{F}}_{k} + \boldsymbol{l}_{k,t};\,\,\,\,\,\,\,\,\, \forall k \in \mathcal{G}\setminus \{i\},
\end{align}
where $\overset{*}{\boldsymbol{F}}_{k}$ is the reference fingerprint for the neighboring objects. There are many ways to find the value of the transforms $\boldsymbol{R}_{k,t}$ and $\boldsymbol{l}_{k,t}$ given the the fingerprint matrices $\boldsymbol{F}_{k,t}$ and the reference fingerprint matrices $\overset{*}{\boldsymbol{F}}_{k}$ of any object $k \in \mathcal{G}\setminus \{i\}$. One such approach is the singular value decomposition (SVD) as follows:
\begin{align}
&[\boldsymbol{U}_{k,t},\boldsymbol{S}_{k,t},\boldsymbol{V}_{k,t}] = g\left(\sum_{j=1}^n (\boldsymbol{f}^{k,t}_j-\boldsymbol{c}_{{\overset{\phantom{*}}{\boldsymbol{F}}}_{k,t}})({\overset{\,\,\,\,\,\,*_{k,t}}{\boldsymbol{f}_j}}-\boldsymbol{c}_{{\overset{*}{\boldsymbol{F}}}_{k}})\!\! \right)\nonumber\\ &\,\,\,\,\,\,\,\,\,\,\,\,\,\,\,\,\,\,\,\,\,\,\,\,\,\,\,\,\,\,\,\,\,\,\,\,\,\,\,\,\,\,\,\,\,\,\,\,\,\,\,\,\,\,\,\,\,\,\,\,\,\,\,\,\,\,\,\,\,\,\,\,\,\,\,\,\,\,\,\,\,\,\,\,\,\,\,\,\,\,\,\,\,\,\,\,\forall k \in \mathcal{G}\setminus \{i\},
\end{align}
where $g(\cdot)$ is defined as the SVD function, $\boldsymbol{U}_{k,t},\boldsymbol{S}_{k,t}$ and $\boldsymbol{V}_{k,t}$ are the factorization matrices, $\boldsymbol{c}_{{\overset{\phantom{*}}{\boldsymbol{F}}}_{k,t}}$ is the mean or centroid of the fingerprint matrix $\boldsymbol{F}_{k,t}$, and $\boldsymbol{c}_{{\overset{*}{\boldsymbol{F}}}_{k}}$ is the mean or centroid of the fingerprint matrix ${\overset{*}{\boldsymbol{F}}}_{k}$. Hence, the rotation $\boldsymbol{R}_{k,t}$ is derived from the factorization matrices as follows:
\begin{align}\label{eqn:ri}
 \boldsymbol{R}_{k,t}= \boldsymbol{V}_{k,t}\, \boldsymbol{U}_{k,t}^\top;\,\,\,\,\,\,\,\,\forall k \in \mathcal{G}\setminus \{i\}.
\end{align}
Meanwhile, the translation $\boldsymbol{l}_{k,t}$ is derived as:
\begin{align}\label{eqn:li}
 \boldsymbol{l}_{k,t}= -\boldsymbol{R}_{k,t} \; \boldsymbol{c}_{{\overset{\phantom{*}}{\boldsymbol{F}}}_{k,t}} + \boldsymbol{c}_{{\overset{*}{\boldsymbol{F}}}_{k}};\,\,\,\,\,\,\,\,\forall k \in \mathcal{G}\setminus \{i\},
\end{align}

The next step is to combine the rotation and translation matrices from all the neighboring objects (i.e. $\boldsymbol{R}_{k,t}$ and $ \boldsymbol{l}_{k,t}; \forall k \in \mathcal{G}\setminus \{i\}$) into a rotation matrix $\boldsymbol{R}_{i,t}$ and a translation matrix $\boldsymbol{l}_{i,t}$ that represent the environment effect on object $i$ given the fact that the environment effect on object $i$ is the same as the environment effect on the neighboring objects. One of the simplest ways to combine the rotation and translation matrices into an estimated rotation matrix and translation matrix is using the minimum mean square error~(MMSE) estimator that minimizes the square of the errors between the estimator  ($\boldsymbol{R}_{i,t}$ and $\boldsymbol{l}_{i,t}$) and all the estimated neighboring matrices ($\boldsymbol{R}_{k,t}$ and $\boldsymbol{l}_{k,t}; \forall k \in \mathcal{G}\setminus \{i\}$). Hence, the MMSE problem is:
\begin{align}\label{eqn:MMSE1}
\underset{\boldsymbol{R}_{i,t}}{\text{argmin}}\;\sum_{k \in \mathcal{G}\setminus \{i\}}\beta_k(\boldsymbol{R}_{i,t}-\boldsymbol{R}_{k,t})^2,\\\label{eqn:MMSE2}
\underset{\boldsymbol{l}_{i,t}}{\text{argmin}}\; \sum_{k \in \mathcal{G}\setminus \{i\}}\beta_k(\boldsymbol{l}_{i,t}-\boldsymbol{l}_{k,t})^2,
\end{align}
where $\beta_k$ is the amount of similarity between object $i$ and object $k$ as shown  in Fig.~\ref{fig:net}. By solving~(\ref{eqn:MMSE1}) and~(\ref{eqn:MMSE2}) for $\boldsymbol{R}_{i,t}$ and $\boldsymbol{l}_{i,t}$ and substituting in~(\ref{eqn:fhat10}) we get the estimated fingerprint matrix $\hat{\boldsymbol{F}}_{i,t}$ as follows:
\begin{align}
\boldsymbol{F}_{i,t}&=\boldsymbol{R}_{i,t}\, \hat{\boldsymbol{F}}_{i,t} + \boldsymbol{l}_{i,t},\\
\hat{\boldsymbol{F}}_{i,t}&=\boldsymbol{R}_{i,t}^{-1}\,\boldsymbol{F}_{i,t}-\boldsymbol{R}_{i,t}^{-1}\,\boldsymbol{l}_{i,t}.\label{eqn:fhat}
\end{align}

Consequently, the distance between the estimated fingerprints $\hat{\boldsymbol{F}}_{i,t}$ and the reference fingerprint matrix becomes:
\begin{align}
\hat{\Delta}_{i,t} = D_B(\hat{\boldsymbol{F}}_{i,t},\overset{*}{\boldsymbol{F}}_{i}),
\end{align}
where $\hat{\Delta}_{i,t}$ represents the estimated distance between the object fingerprints and the reference fingerprints after removing the effect of the environment. To avoid the matrix inversion in~(\ref{eqn:fhat}), the environment effect can be applied to the reference fingerprint $\overset{*}{\boldsymbol{F}}_{i}$ as follows:
\begin{align}
\overset{**}{\boldsymbol{F}}_{i}=\boldsymbol{R}_{i,t}\, \overset{*}{\boldsymbol{F}}_{i} + \boldsymbol{l}_{i,t}
\end{align}
where $\overset{**}{\boldsymbol{F}}_{i}$ is the updated fingerprint reference matrix for object $1$ that includes the effect of environment. Thus, the fingerprint matrix $\boldsymbol{F}_{i,t}$ can be compared with  $\overset{**}{\boldsymbol{F}}_{i}$ since both matrices include the effect of the environment. The estimated distance in this case will be as follows:
\begin{align}\label{eqn:star}
\hat{\Delta}_{i,t} = D_B(\boldsymbol{F}_{i,t},\overset{**}{\boldsymbol{F}}_{i}).
\end{align}

 To show the advantage of removing the effect of environment by using the estimated distance $\hat{\Delta}_{i,t}$ instead of the distance $\Delta_{i,t}$, we consider the two attack scenarios introduced in Section~\ref{sec2} and compare between the environment estimation case and the no environment estimation case.

In the cyber emulation attack case, the fingerprints of the malicious messages sent from the attacker represent the attacker's device fingerprints $\boldsymbol{F}_{m,t}$. Therefore, using classical device fingerprinting techniques, the fingerprint of the attacker $\boldsymbol{F}_{m,t}$ is compared with the reference fingerprint ${\overset{*}{\boldsymbol{F}}}_{k}$ for the target object as in: 
\begin{align}
\Delta_{m,t} = D_B(\boldsymbol{F}_{m,t},{\overset{*}{\boldsymbol{F}}}_{i}),
\end{align}
Based on the value of $\Delta_{m,t}$ we differentiate between two cases: 
\begin{itemize}
\item $\Delta_{m,t}\leq\tau$: the attack is successful. Malicious messages sent by the attacker are treated as messages  sent from the target object.
\item $\Delta_{m,t}>\tau$: the attack is not successful. Malicious messages sent by the attacker are different from the messages sent from the target object and attack flag is raised.  
\end{itemize}

Conventional device fingerprinting techniques choose large values of the threshold $\tau$ to avoid false positives when the environmental effects on legitimate objects are enough to make $\Delta_{i,t}$ larger than the threshold $\tau$ as shown in \cite{bertoncini2012wavelet}. However, large values of $\tau$ increase the chance of a successful attack since $\Delta_{m,t}<\tau$ for a successful attack. Thus, the estimation of the expected environmental effects on legitimate objects is essential to reduce the value of the threshold $\tau$ without increasing the false positives. In the following theorem, we derive a closed-form solution which proves that environment estimation during object fingerprinting in the IoT improves detection rate without increasing the false positives rate. In this theorem, we consider the distributions of the fingerprints to be Gaussian in order to use the closed-form expression of the Bhattacharyya distance. However, in the evaluation section, we show that the result will still hold for any distribution of fingerprints.
\begin{theorem}
Environment estimation during object fingerprinting in the IoT improves detection rate without increasing the false positives rate.
\end{theorem}
\begin{proof}
From the hypothesis in~(\ref{eqn:hyp}), the detection rate of attackers can be represented as:
\begin{align}
P_{\textit{D}}&=P\big\{\Delta_{i,t}>\tau\big|\mathcal{H}_1\big\}.
\end{align}
When the environment effect on objects is not considered, the detection threshold $\tau'$ should satisfy the following condition:
\begin{align}
\tau' \geq D_B(\boldsymbol{F}_{i,t},\overset{*}{\boldsymbol{F}}_{i}).
\end{align}
This condition represents the best case scenario, which is when the collected fingerprints are generated from the legitimate object. This condition is necessary to avoid false positives. The collected fingerprints from ~(\ref{eqn:fhat}) can be represented as:
\begin{align}
\boldsymbol{F}_{i,t}=\boldsymbol{R}_{i,t}\, \overset{*}{\boldsymbol{F}}_{i} + \boldsymbol{l}_{i,t}+\boldsymbol{w},
\end{align}
where $\overset{*}{\boldsymbol{F}}_{i}$ is the reference fingerprint, $\boldsymbol{R}_{i,t}$ and $\boldsymbol{l}_{i,t}$ are the rotation and translation transformation which capture the effect of the environment, and $\boldsymbol{w}$ is the noise. The threshold $\tau'$ should be larger than the effect of noise and the effect of environment $(\boldsymbol{R}_{i,t},\boldsymbol{l}_{i,t})$ on the legitimate object, hence $\tau'$ is related to $\boldsymbol{R}_{i,t}$ and $\boldsymbol{l}_{i,t}$. The lower bound on the threshold $\tau'$ is:
\begin{align}
\tau'_{LB} &= D_B({ \boldsymbol{F}_{i,t},\overset{*}{\boldsymbol{F}}_{i}})\nonumber\\
&=\frac{1}{8}(\boldsymbol{\mu}_{\scriptstyle \boldsymbol{F}_{i,t}}-\boldsymbol{\mu}_{\scriptstyle{\overset{*}{\boldsymbol{F}}_{i}}})^T\boldsymbol{\Sigma}^{-1}(\boldsymbol{\mu}_{\scriptstyle \boldsymbol{F}_{i,t}}-\boldsymbol{\mu}_{\scriptstyle{\overset{*}{\boldsymbol{F}}_{i}}})\nonumber\\
&\qquad\qquad\qquad+\frac{1}{2}\ln{\frac{\det\boldsymbol{\Sigma}}{\sqrt{\det\boldsymbol{\Sigma}_{\scriptstyle \boldsymbol{F}_{i,t}}\det\boldsymbol{\Sigma}_{\scriptstyle{\overset{*}{\boldsymbol{F}}_{i}}}}}},\nonumber\\
&=\frac{1}{8}(\boldsymbol{\mu}_{\scriptstyle \boldsymbol{F}_{i,t}}-\boldsymbol{\mu}_{\scriptstyle{\overset{*}{\boldsymbol{F}}_{i}}})^T\boldsymbol{\Sigma}^{-1}(\boldsymbol{\mu}_{\scriptstyle \boldsymbol{F}_{i,t}}-\boldsymbol{\mu}_{\scriptstyle{\overset{*}{\boldsymbol{F}}_{i}}}),\nonumber\\
&=\frac{1}{8}(\boldsymbol{c}_{{\overset{\phantom{*}}{\boldsymbol{F}}}_{i,t}}-\boldsymbol{c}_{{\overset{*}{\boldsymbol{F}}}_{i}})^T\boldsymbol{\Sigma}^{-1}(\boldsymbol{c}_{{\overset{\phantom{*}}{\boldsymbol{F}}}_{i,t}}-\boldsymbol{c}_{{\overset{*}{\boldsymbol{F}}}_{i}}),
\end{align}
%\begin{align}
%\tau'_{LB} = D_B(\boldsymbol{F}_{i,t},\overset{*}{\boldsymbol{F}}_{i})&=\frac{1}{4}\ln\left(\frac{1}{4}\left(\frac{\boldsymbol{\Sigma}_1^2}{\boldsymbol{\Sigma}_2^2}+\frac{\boldsymbol{\Sigma}_2^2}{\boldsymbol{\Sigma}_1^2}+2\right)\right)\notag\\
%&\qquad\qquad+\frac{1}{4}\left(\frac{(\boldsymbol{\mu}_1-\boldsymbol{\mu}_2)^2}{\boldsymbol{\Sigma}_1^2+\boldsymbol{\Sigma}_2^2}\right),\notag\\
%&=\frac{1}{8}\frac{(\boldsymbol{\mu}_1-\boldsymbol{\mu}_2)^2}{\boldsymbol{\Sigma}^2},\notag\\
%&=\frac{1}{8}\frac{\left(\boldsymbol{c}_{{\overset{\phantom{*}}{\boldsymbol{F}}}_{i,t}}-\boldsymbol{c}_{{\overset{*}{\boldsymbol{F}}}_{i}}\right)^2}{\boldsymbol{\Sigma}^2},
%\end{align}
where  $\boldsymbol{\mu}_{\scriptstyle \boldsymbol{F}_{i,t}}, \boldsymbol{\mu}_{\scriptstyle{\overset{*}{\boldsymbol{F}}_{i}}},\boldsymbol{\Sigma}_{\scriptstyle \boldsymbol{F}_{i,t}}$ and $\boldsymbol{\Sigma}_{\scriptstyle{\overset{*}{\boldsymbol{F}}_{i}}}$ are the means and covariances for fingerprints $\boldsymbol{F}_{i,t}$ and $\overset{*}{\boldsymbol{F}}_{i}$ respectively, $\boldsymbol{\Sigma}=(\boldsymbol{\Sigma}_{\scriptstyle_{\boldsymbol{F}_{i,t}}}+\boldsymbol{\Sigma}_{\scriptstyle_{\overset{*}{\boldsymbol{F}}_{i}}}) / 2$, and $\boldsymbol{c}_{{\overset{\phantom{*}}{\boldsymbol{F}}}_{i,t}}$ and $\boldsymbol{c}_{{\overset{*}{\boldsymbol{F}}}_{i}}$ are the centroids of fingerprint matrices $\boldsymbol{F}_{i,t}$ and $\overset{*}{\boldsymbol{F}}_{i}$ respectively, hence:
\begin{align}
\tau'_{LB} &=\frac{1}{8}(\boldsymbol{c}_{{\overset{*}{\boldsymbol{F}}}_{i}}\!\!+\boldsymbol{l}_{i,t}+\boldsymbol{w}-\boldsymbol{c}_{{\overset{*}{\boldsymbol{F}}}_{i}})^T\boldsymbol{\Sigma}^{-1}(\boldsymbol{c}_{{\overset{*}{\boldsymbol{F}}}_{i}}\!\!+\boldsymbol{l}_{i,t}+\boldsymbol{w}-\boldsymbol{c}_{{\overset{*}{\boldsymbol{F}}}_{i}}),\notag\\
&=\frac{1}{8}(\boldsymbol{l}_{i,t}+\boldsymbol{w})^T\boldsymbol{\Sigma}^{-1}(\boldsymbol{l}_{i,t}+\boldsymbol{w}).\label{eqn:lowb}
\end{align}
Therefore, if the environment effect on objects $\boldsymbol{l}_{i,t}$ is high, the detection threshold $\tau'$ has to be increased to avoid false positives. However, by estimating the environmental effects on objects and adding these environmental effects to the reference fingerprint as in~(\ref{eqn:star}), the detection threshold becomes:
\begin{align}
\tau &\geq D_B(\boldsymbol{F}_{i,t},\overset{**}{\boldsymbol{F}}_{i}),\\
\tau &\geq D_B\Bigl((\boldsymbol{R}_{i,t}\, \overset{*}{\boldsymbol{F}}_{i} + \boldsymbol{l}_{i,t}+\boldsymbol{w}),(\boldsymbol{R}_{i,t}\, \overset{*}{\boldsymbol{F}}_{i} + \boldsymbol{l}_{i,t})\Bigr).
\end{align}
Similarly, when the received fingerprints belong to the legitimate source object, the lower bound on the threshold $\tau$ is:
\begin{align}
\tau_{LB} &= D_B(\boldsymbol{F}_{i,t},\overset{**}{\boldsymbol{F}}_{i}),\notag\\
&=\frac{1}{8}(\boldsymbol{c}_{{\overset{*}{\boldsymbol{F}}}_{i}}\!\!+\boldsymbol{l}_{i,t}+\boldsymbol{w}-\boldsymbol{c}_{{\overset{*}{\boldsymbol{F}}}_{i}}-\boldsymbol{l}_{i,t})^T\boldsymbol{\Sigma}^{-1}\notag\\
&\qquad\qquad\qquad\qquad(\boldsymbol{c}_{{\overset{*}{\boldsymbol{F}}}_{i}}\!\!+\boldsymbol{l}_{i,t}+\boldsymbol{w}-\boldsymbol{c}_{{\overset{*}{\boldsymbol{F}}}_{i}}-\boldsymbol{l}_{i,t}),\notag\\
&=\frac{1}{8}\boldsymbol{w}^T\boldsymbol{\Sigma}^{-1}\boldsymbol{w} < \tau'_{LB}.\label{eqn:lowbe}
\end{align}
Hence, the threshold lower bound $\tau_{LB}$ when the environment effects are estimated is only related to noise $\boldsymbol{w}$ and is lower than the threshold lower bound $\tau'_{LB}$ when the environment effects are not estimated. Therefore, estimating the environment allows us to pick a lower value for the threshold while keeping the same false positives rate. The ability to pick lower values for the threshold improves the detection rate since attackers are required to replicate the exact fingerprints of a legitimate object  and small changes compared with the legitimate objects will cause the estimation to be over the threshold.
\end{proof}

In the cyber-physical emulation attack case, the attacker can generate fingerprints $\boldsymbol{F}_{m,t}$ which are identical to the fingerprints generated from the attacked object.
%  the fingerprint of the attacker's messages $\boldsymbol{F}_{m,t}$ is:
%\begin{align}\label{eqn:27}
%M_A=F_S+w+E_A,
%\end{align}
%where the attacker can generate the fingerprint $F_S$ of the original legitimate object $S$. 
Classical device fingerprinting techniques, such as in~\cite{7980167}, fail completely to detect such type of attacks, since classical techniques rely only on the reference fingerprint $\overset{*}{\boldsymbol{F}}_{i}$ to compare it with fingerprint $\boldsymbol{F}_{m,t}$. However, estimating the environment effect on objects requires the attacker to replicate both the legitimate object fingerprint $\boldsymbol{F}_{i,t}$ and the environment effect on the legitimate object at each time interval $t$. Next, we show that classical fingerprinting techniques cannot detect any cyber-physical emulation attackers while environmental estimation framework can detect such an attack. Similar to Theorem 1, we consider the distributions of the fingerprints to be Gaussian in order to use the closed-form expression of the Bhattacharyya distance. However, in the evaluation section,  we will see that the results will hold for any distribution.
\begin{theorem}
Classical fingerprinting techniques in the IoT cannot detect any cyber-physical emulation attackers.
\end{theorem}
\begin{proof}
The detection rate of attackers as shown in (\ref{eqn:hyp}) is:
\begin{align}
P_{\textit{D}}&=P\big\{\Delta f>\tau\big|\mathcal{H}_1\big\}.
\end{align}
In replication attacks, the attacker generates fingerprints $\boldsymbol{F}_{m,t}$ where:
\begin{align}
\boldsymbol{F}_{m,t}=\boldsymbol{R}_{m,t}\, \overset{*}{\boldsymbol{F}}_{i} + \boldsymbol{l}_{m,t}+\boldsymbol{w}.
\end{align}
Without estimating the environmental effect, the attack is successful when:
\begin{align}
D_B(\boldsymbol{F}_{m,t},{\overset{*}{\boldsymbol{F}}}_{i})\leq\tau'.
\end{align}
%For tractability, we assume normal distributions for $\boldsymbol{F}_{m,t}$ and $\overset{*}{\boldsymbol{F}}_{i}$ with the same variance.  This assumption is made to use the closed-form expression for the Bhattacharyya distance. 
Using the lower bound on the detection threshold $\tau'$ in~(\ref{eqn:lowb}), the condition for successful attack becomes as follows:
\begin{align}
&\qquad D_B(\boldsymbol{F}_{m,t},{\overset{*}{\boldsymbol{F}}}_{i})\leq\frac{1}{8}(\boldsymbol{l}_{i,t}+\boldsymbol{w})^T\boldsymbol{\Sigma}^{-1}(\boldsymbol{l}_{i,t}+\boldsymbol{w}),\notag\\
&\frac{1}{8}(\boldsymbol{c}_{{\overset{*}{\boldsymbol{F}}}_{i}}\!\!\!\!+\boldsymbol{l}_{m,t}\!\!+\boldsymbol{w}\!-\boldsymbol{c}_{{\overset{*}{\boldsymbol{F}}}_{i}}\!\!)^T\boldsymbol{\Sigma}^{-1}(\boldsymbol{c}_{{\overset{*}{\boldsymbol{F}}}_{i}}\!\!\!\!+\boldsymbol{l}_{m,t}\!\!+\boldsymbol{w}\!-\boldsymbol{c}_{{\overset{*}{\boldsymbol{F}}}_{i}}\!\!)\leq\notag\\
&\qquad\qquad\qquad\qquad\qquad\qquad\frac{1}{8}(\boldsymbol{l}_{i,t}+\boldsymbol{w})^T\boldsymbol{\Sigma}^{-1}(\boldsymbol{l}_{i,t}+\boldsymbol{w}),\notag\\
&\frac{1}{8}(\boldsymbol{l}_{m,t}\!\!+\boldsymbol{w}\!)^T\boldsymbol{\Sigma}^{-1}(\boldsymbol{l}_{m,t}\!\!+\boldsymbol{w}\!)\leq\frac{1}{8}(\boldsymbol{l}_{i,t}+\boldsymbol{w})^T\boldsymbol{\Sigma}^{-1}(\boldsymbol{l}_{i,t}+\boldsymbol{w}).
\end{align}
%\begin{align}
%D_B\Bigl(f_1(F_S+w+E_A),f_2(F_S^*)\Bigr)<\frac{1}{8}\frac{(\mu_{w}+\mu_{e_S})^2}{\sigma^2},\\
%\frac{1}{8}\frac{(\mu_{w}+\mu_{E_A})^2}{\sigma^2}<\frac{1}{8}\frac{(\mu_{w}+\mu_{e_S})^2}{\sigma^2},
%\end{align}
Hence, the attack is successful if the environment effect on the attacker's device $\boldsymbol{l}_{m,t}$ is not significantly larger than the environment effect on the legitimate object $\boldsymbol{l}_{i,t}$. On the other hand, by estimating the environment effects on fingerprints, the attack is only successful when:
\begin{align}
D_B(\boldsymbol{F}_{m,t},{\overset{**}{\boldsymbol{F}}}_{i})\leq\tau,
\end{align}
where $\tau$ is smaller than $\tau'$.  
%\begin{align}
%D_B\Bigl(f_1(F_S+w+E_A),f_2(F_S^*+e_S^*)\Bigr)<\tau,
%\end{align}
By substituting the detection threshold $\tau$ from~(\ref{eqn:lowbe}), we have:
\begin{align}
&D_B(\boldsymbol{F}_{m,t},{\overset{**}{\boldsymbol{F}}}_{i})\leq\frac{1}{8}\boldsymbol{w}^T\boldsymbol{\Sigma}^{-1}\boldsymbol{w},\\
\frac{1}{8}(\boldsymbol{c}_{{\overset{*}{\boldsymbol{F}}}_{i}}\!\!\!\!+\boldsymbol{l}_{m,t}\!\!+\boldsymbol{w}-&\boldsymbol{c}_{{\overset{*}{\boldsymbol{F}}}_{i}}\!\!\!\!-\boldsymbol{l}_{i,t})^T\boldsymbol{\Sigma}^{-1}\notag\\
&(\boldsymbol{c}_{{\overset{*}{\boldsymbol{F}}}_{i}}\!\!\!\!+\boldsymbol{l}_{m,t}\!\!+\boldsymbol{w}-\boldsymbol{c}_{{\overset{*}{\boldsymbol{F}}}_{i}}\!\!\!\!-\boldsymbol{l}_{i,t})\leq\frac{1}{8}\boldsymbol{w}^T\boldsymbol{\Sigma}^{-1}\boldsymbol{w},\\
\frac{1}{8}(\boldsymbol{l}_{m,t}\!\!+\boldsymbol{w}-\boldsymbol{l}_{i,t})^T&\boldsymbol{\Sigma}^{-1}
(\boldsymbol{l}_{m,t}\!\!+\boldsymbol{w}-\boldsymbol{l}_{i,t})\leq\frac{1}{8}\boldsymbol{w}^T\boldsymbol{\Sigma}^{-1}\boldsymbol{w}
\end{align}
In this case, unless the environment effect on the attacker's device $\boldsymbol{l}_{m,t}$ is identical to the environment effect on the legitimate object $\boldsymbol{l}_{i,t}$ the attack is not successful. Therefore, the attacker is required to be physically in the same environment as the legitimate object which makes remote cyber-physical emulation attacks detectable.
%the replication attack is detected when $|\mu_{E_A}-\nobreak\mu_{e_S^*}|\ne\nobreak 0$.
\end{proof}
Therefore, by estimating the environment, our proposed approach is able to determine whether an object is present in the same environment as the surrounding objects or whether the object is located in a remote area. Additionally, the estimation process involves using historical data to determine which objects are considered as surrounding objects in the estimation process. Hence, attackers, such as cyber-physical emulation attackers, are unable to replicate the environment effects in their remote location which allows our approach to detect such attacks.

Objects in the IoT system have different types of features which requires our proposed approach to use a different feature space for each type of objects. In the following section, we introduce a transfer learning approach to allow our environment estimation to use objects of different types and objects that have different feature spaces.

\vspace{0.1cm}
\section{Transfer Learning for Environment Estimation}\label{sec:tranlern}
Estimating the environmental effects using only similarity measures requires all objects in the IoT system to have the same feature space, i.e. all fingerprints $\boldsymbol{f}^{i,t}_x$ in the fingerprint matrix $\boldsymbol{F}_{i,t}$ have the same number of features. The condition to have the same feature space for all fingerprints has many disadvantages. First, the number of objects that share the same environmental effects and have the same feature space can often be relatively small which reduces the size of the training data. Second, having fewer number of objects with the same feature space decreases the estimation accuracy of the environmental effects since estimation process involves MMSE estimator. Third, attackers can compromise the estimation of an object by attacking all the nearby objects that have the same feature space, i.e. objects of the same type.
%which makes it difficult to estimate environmental effects affecting multiple objects. which are directly related to a target object, the target object in our case is the source object $S$, is relatively small compared to the total number of related objects regardless of the feature space. Also, estimating environmental effects based on objects of the same type only makes it more vulnerable to attacks. Instead of attacking the target object $x$ only, the attacker performs the attack on all objects of the same type as $x$ in the area surrounding the target object $x$. 
Therefore, using objects of different types in the estimation process enhances the performance and increases the complexity and costs for attackers to determine which objects are used in the estimation process. One of the tools to combine multiple problems with different feature spaces is called \emph{transfer learning}~\cite{pan2010survey}.
%However, some objects of types different than the target object $x$ can also be related to object $x$ and endure the same environmental effects on $x$. To better estimate the environmental effects on $x$, the knowledge from auxiliary objects, objects related to $x$ but with types different than $x$, can be used. 
Transfer learning is a tool to transfer knowledge gained while solving one task to improve the learning of a different but related task, where the primary task is referred to as target task and the related task as source task.

In environmental effect estimation, we denote the fingerprints of objects with the same feature space as target data $\boldsymbol{F}_{k,t};\forall k\in \mathcal{D_T}$ and the objects with different feature space but sharing the same environmental effects as source data $\boldsymbol{F}_{k,t}^s;\forall k\in \mathcal{D_S}$, where $\mathcal{D_T}$ is the set of target data and $\mathcal{D_S}$ is the set of source data. Additionally, we denote the target task as to estimate the transformations $\boldsymbol{R}_{i,t}$ and $\boldsymbol{l}_{i,t}$ for object $i$ given the target data $\boldsymbol{F}_{k,t};\forall k\in \mathcal{D_T}$. The transfer learning procedure consists of two steps. In the first step, using the source fingerprints $\boldsymbol{F}_{k,t}^s$ and following the same steps in Section~\ref{sec:estenv} we determine the transformation matrices $\boldsymbol{R}_{k,t}^s$ and $\boldsymbol{l}_{k,t}^s$ similar to~(\ref{eqn:ri}) and~(\ref{eqn:li}) where:
\begin{align}
\boldsymbol{F}_{k,t}^s=\boldsymbol{R}_{k,t}^s\, \overset{*}{\boldsymbol{F}^s_{k}} + \boldsymbol{l}_{k,t}^s;\,\,\,\,\,\,\,\,\, \forall k \in \mathcal{D_S}.\vspace{0.5cm}
\end{align}
Additionally, using the target fingerprints  $\boldsymbol{F}_{k,t}$ and using steps in Section~\ref{sec:estenv}, we get the transformation matrices $\boldsymbol{R}_{k,t}$ and $\boldsymbol{l}_{k,t}$ shown in~(\ref{eqn:ri}) and~(\ref{eqn:li}).
In the second step, we formulate an optimization problem to combine the transformations from the target data $\boldsymbol{R}_{k,t}$ and $\boldsymbol{l}_{k,t}$ and the transformations from the source data $\boldsymbol{R}_{k,t}^s$ and $\boldsymbol{l}_{k,t}^s$ as follows:
\begin{align}\label{eqn:MMSE1tf}\vspace{0.5cm}
\underset{\boldsymbol{R}_{i,t}}{\text{argmin}}\;\sum_{k\in D_T}\beta_k(\boldsymbol{R}_{i,t}-\boldsymbol{R}_{k,t})^2+\alpha\sum_{k\in D_S}\beta_k(\boldsymbol{R}_{i,t}-\boldsymbol{R}_{k,t}^s)^2,\\\label{eqn:MMSE2tf}
\underset{\boldsymbol{l}_{i,t}}{\text{argmin}}\;\sum_{k\in D_T}\beta_k(\boldsymbol{l}_{i,t}-\boldsymbol{l}_{k,t})^2+\alpha\sum_{k\in D_S}\beta_k(\boldsymbol{l}_{i,t}-\boldsymbol{l}_{k,t}^s)^2,
\end{align}
where $\alpha$ is the weight of transfer which determines the amount of effect the source data has on the target task. 

The solution for the joint learning formulation in~(\ref{eqn:MMSE1tf}) and~(\ref{eqn:MMSE2tf}) is the transformation matrices $\boldsymbol{R}_{i,t}$ and $\boldsymbol{l}_{i,t}$ that represent the environmental effects on object $i$ giving the target fingerprints and source fingerprints. The objective functions in~(\ref{eqn:MMSE1tf}) and~(\ref{eqn:MMSE2tf}) follow the form of joint convex optimization that allows us to use a range of efficient convex optimization algorithms to solve it such as standard gradient methods~\cite{bertsekas1999nonlinear}.
The limitation of previous transfer learning method is that it requires objects in both the target and source domains to have similar transformations. However, even if transformations are different in the target and source domains, other types of knowledge can be transferred as well between objects such as the value of the threshold $\tau$. For example, if objects in the target and source domains have different transformations but have similar estimation errors, such as same type of noise, or same data collection errors. In this case, the value of the threshold $\tau$ can be estimated from the source domain and used on the target domain, since the target domain has fewer data points than source domain to properly estimate $\tau$. 

The proposed transfer learning approach enables the environment estimation algorithm to use different types of objects in the estimation process. The transfer learning approach is essential in practical IoT systems which have wide diversity of objects with different types of features.

\vspace{0.1cm}
\section{Evaluation and Results}\label{sec:results}
To evaluate our approach, we used RFID data collected by Bertoncini et al. in~\cite{bertoncini2012wavelet}. The data consists of IQ recordings of 25 RFID tags of model AD (Avery-Dennison AD612) purchased from the same manufacturer. The writing and reading of RFID tags was performed with a Thing Magic Mercury 5e RFID Reader, and the antenna was an omnidirectional antenna from Laird technologies the sampling frequency was 4.0 Msps. The same Electronic Product Code (EPC) was written on all RFID tags to insure all RFID tags are identical. The recordings were captured at three different physical conditions: The first is at normal room temperature, the second RFID tags were warped by water submersion, and lastly RFID tags were warped by physically changing the angle between the RFID tag and the reader. The use of RFID technology in supply chain management is one of the primary applications of IoT systems. For example, Wal-Mart uses the RFID tags to manage its supplies~\cite{dillman2003wal}. In the supply chain management applications, RFID tags are scanned as the supplies move under the RFID tag reader which is similar to the third physical condition of the RFID dataset that we have. Therefore, in our evaluation, we will use the change of angle between reader and the RFID tag as an environmental effect that affects all the objects. 
\begin{figure}[!t]
  \begin{center}
	\centering
 	   \includegraphics[width=8cm]{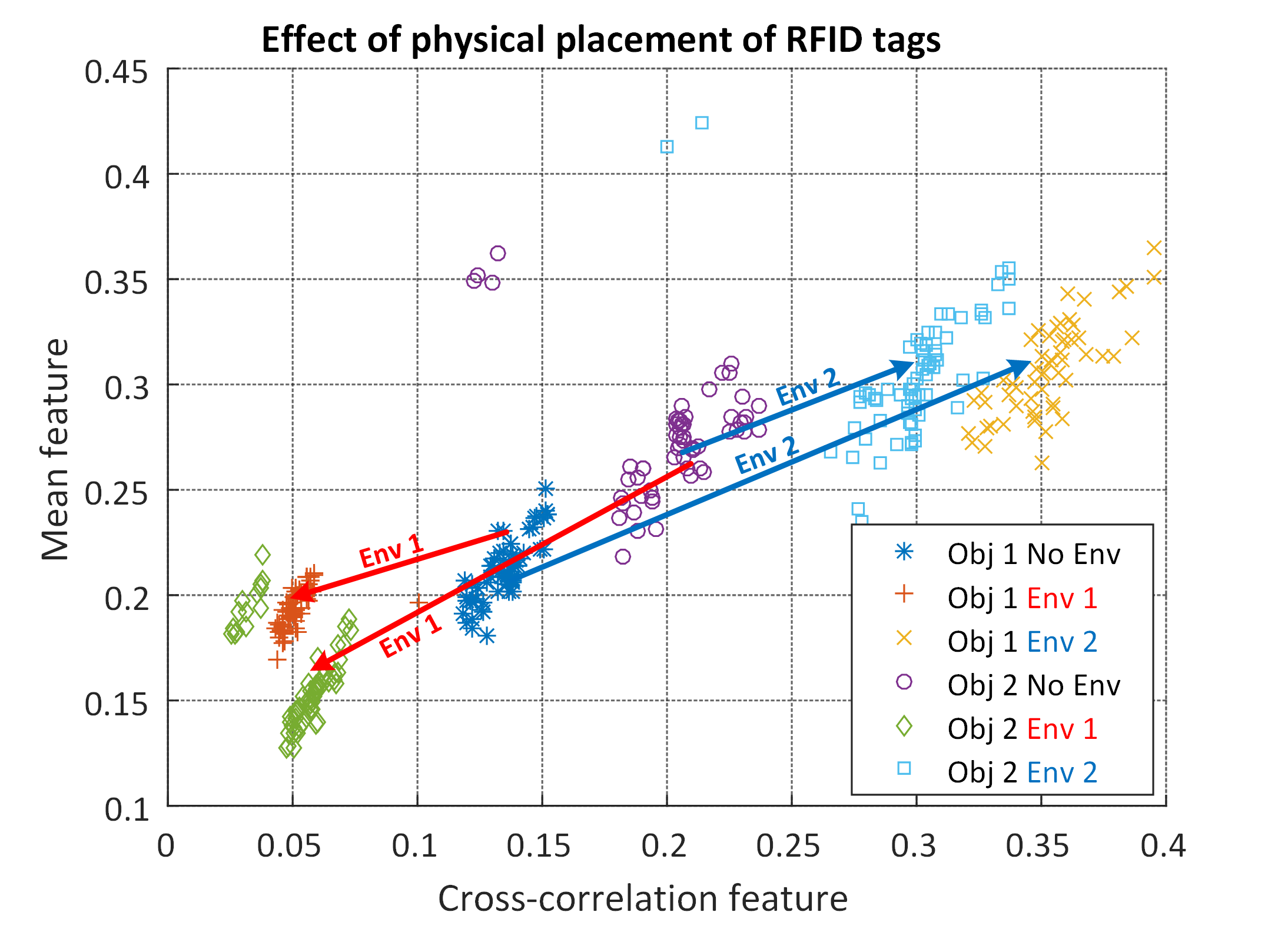}
  	  \caption{\small{\label{fig:feat}Effect of changing the angle between RFID tags and the RFID reader on two features of the RFID fingerprints which are the mean and cross-correlation features.}}
  \end{center}
\end{figure} 

\subsection{Feature Extraction}
For each RFID tag, we extract seven features from the IQ recordings as described in~\cite{bertoncini2012wavelet}: Mean of the EPC, variance of EPC, Shannon entropy, second central moment, skewness, kurtosis, and maximum cross-correlation. These features are referred to as higher order statistical calculations. In Fig.~\ref{fig:feat}, the mean of the EPC and the cross-correlation features are shown for two RFID tags. The figure shows that both RFID tags showed similar changes when the angle between the RFID tag and the reader changes. Additionally, from Fig.~\ref{fig:feat}, we notice that the effect of changing the angle is enough to shift the fingerprints features of an object and make these fingerprints closer to a completely different object than to the same object but without any environmental effects. For example, in Fig.~\ref{fig:feat}, the effect of the environment, highlighted using blue arrows with caption \emph{Env 2}, on the first object shifts the fingerprints closer to the second object than to the original fingerprints of the first object. Thus, the fingerprints of the first object would have been mistakenly assigned to the second object if the method of assignment relied only on the distance to the original object fingerprints without the use of any environmental estimation techniques.

\subsection{Evaluation Results}
The first evaluation for our approach is to show the effect of environment estimation on the accuracy of assigning each fingerprint to the correct object. Fig.~\ref{fig:acc} shows the average distance between collected fingerprints and the reference fingerprints. In the case of traditional algorithms, the existence of environmental effects on objects significantly increases the average distance between collected fingerprints and reference fingerprints. However, using the proposed algorithm to remove the environmental effects drops the average distance to values closer to the no environmental effects case. Fig.~\ref{fig:acc} shows close to 90\% improvement in average for total objects between 4 and 15 objects. The impact of having lower average distance between collected fingerprints and the reference fingerprints on the object authentication depends on the choice of detection threshold.

Fig.~\ref{fig:thre} shows the impact of the threshold on the percentage of correctly assigned fingerprints. In this experiment, the total number of objects is 20. In the case where environmental effects are not estimated, the minimum possible threshold to achieve 50 percent accurate assignment of objects. however, estimating the environment allows us to drop the value of the threshold to 3 and achieve a 100\% accurate assignment of the objects. The benefit of having low values for the threshold is shown in the following attack scenarios. 
\begin{figure}[!t]
  \begin{center}
	\centering
 	   \includegraphics[width=8cm]{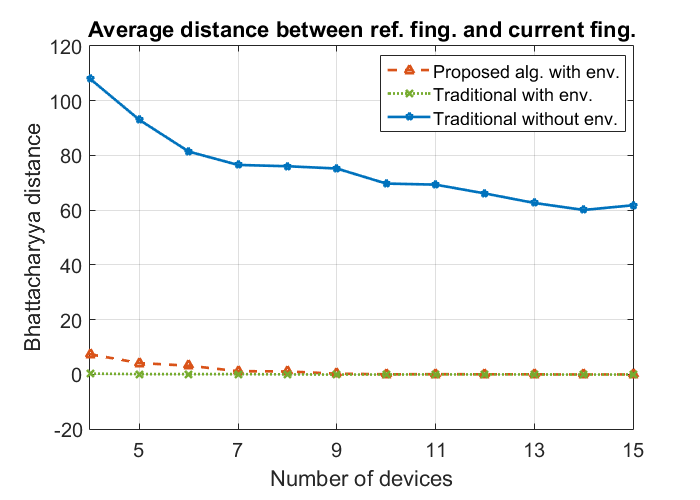}
  	  \caption{\small{\label{fig:acc}Average distance between reference fingerprints and current fingerprints with and without environment estimation with respect to the number of objects.}}
  \end{center}
\end{figure}

\begin{figure}[!t]
  \begin{center}\vspace{-0.5cm}
	\centering
 	   \includegraphics[width=8cm]{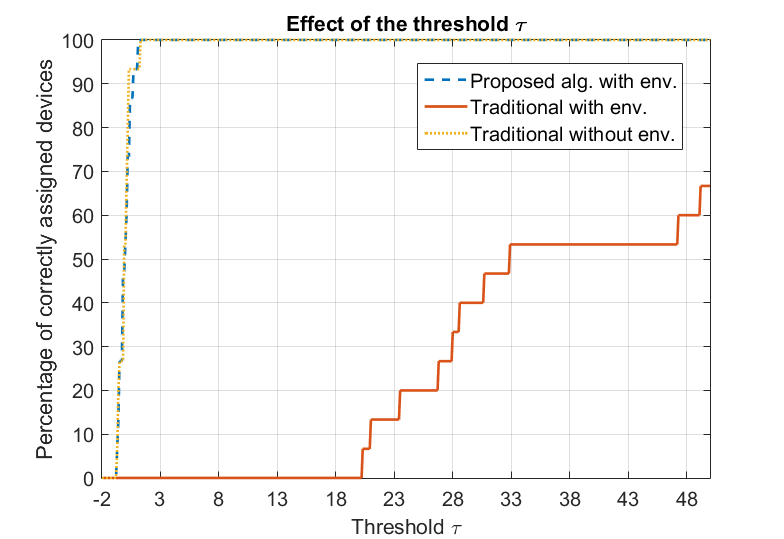}
  	  \caption{\small{\label{fig:thre}Effect of the threshold $\tau$  comparison between the proposed approach with and without transfer learning with respect to the number of objects.}}
  \end{center}
\end{figure}

In Fig.~\ref{fig:norep} we simulate a cyber emulation attack scenario. In this scenario, we have one attacker object throughout the experiment which explains the similar Bhattacharyya distance for the attacker object as the total number of objects in the experiment increases. From Fig.~\ref{fig:norep}, we can see that estimating the environmental effects increases the gap between the legitimate objects assignment and the attacker objects assignment. This large gap allows the choice of any value for the threshold between 10 and 150 to perfectly detect the attacker. However, in the traditional approach, the gap between the legitimate objects and the attacker object is between 70 and 150 which is smaller than the gap in the proposed approach. Fig.~\ref{fig:norep} shows that the proposed approach improves the gap of the threshold by $40\%$.
%The increase in the gap of the threshold due to the estimation of environment in the proposed approach equals to 40\% improvement.
This result, corroborates Theorem 1, as it shows that estimating the effects of the environment allows the use of a lower threshold while ensuring that attackers are easier to detect as shown in Fig.~\ref{fig:norep}. Therefore, for cases in which the total number of objects is more than seven objects, the improvements translate to 100\% over the traditional approach that were unable to detect the attack.
\begin{figure}[!t]
  \begin{center}
	\centering
 	   \includegraphics[width=8cm]{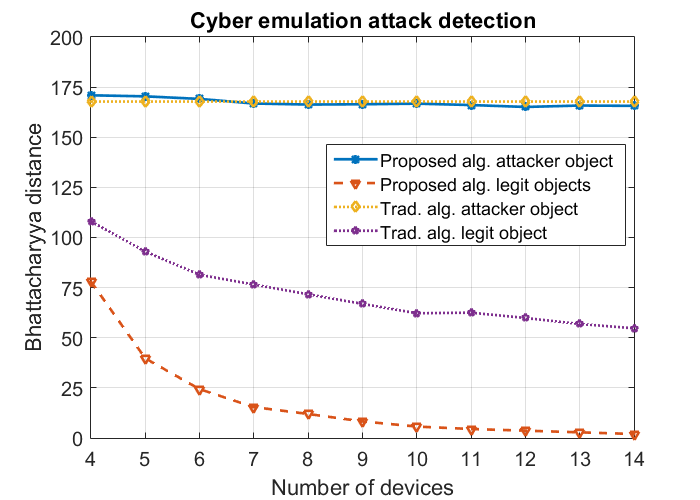}
  	  \caption{\small{\label{fig:norep}Average distance between reference fingerprints and current fingerprints in the case of a cyber emulation attack.}}
  \end{center}
\end{figure}

In Fig.~\ref{fig:repp}, we study the cyber-physical emulation attack scenario. In this figure, we used the same object fingerprints as the malicious fingerprints but under different environmental effects and similar to the previous experiment we use one attacker object as the total number of object increases. From Fig. 8, we can see that the traditional approach is unable to detect the attacker and assigns the attacker object as a legitimate object for any choice of threshold since the attacker's distance is always below the legitimate object's minimum distance. Meanwhile, for threshold values between 10 and 25, the proposed algorithm is able to detect the attacker when the total number of objects is more than seven, i.e. six legitimate objects and one attacker object. This result, corroborates Theorem 2, as it shows that environment estimation is able to detect cyber-physical emulation attacks that traditional algorithms were unable to detect.
In both emulation attack scenarios, we can observe that the average distance of the legitimate objects is significantly high when the total number of objects is small, i.e., less than seven objects, compared to when the number of objects is big, i.e., more than seven objects. This is due to the effect the attacker has on the estimation process. During the estimation process of legitimate objects, the attacker object is considered as a neighbor object in the estimation process and hence its effect is significant when the total number of objects is small. In contrast, for cases in which the total number of objects is large, the effect the attacker object has on the environment estimation process of all the legitimate objects in the system is small. This limitation can be overcome by implementing a multi-stage estimation process, where multiple estimation processes are executed at different time frames. At each time frame, all objects labeled as attackers are excluded from future estimation processes. This approach ensures that all objects in the estimation process are legitimate.

\begin{figure}[!t]
  \begin{center}
	\centering
 	   \includegraphics[width=8cm]{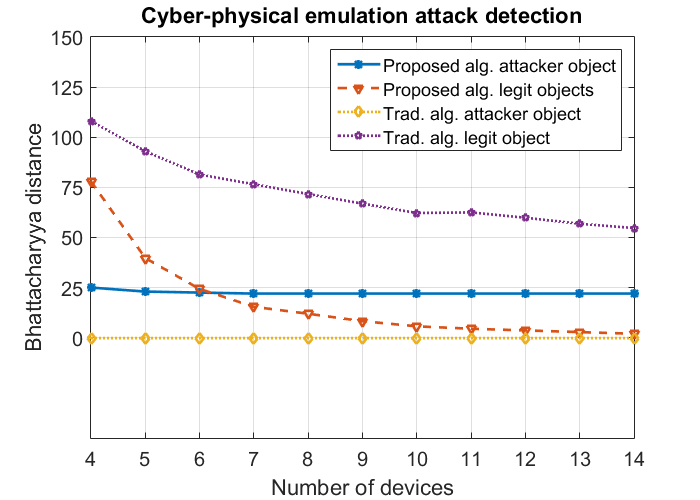}
  	  \caption{\small{\label{fig:repp}Average distance between reference fingerprints and current fingerprints in the case of a cyber-physical emulation attack detection. The attack is successful, when the attacker object distance is below the legitimate object distance.}}
  \end{center}
\end{figure}

To simulate the effect of transfer learning, first we divided the dataset into two groups of objects, a target group with five objects and a source group with 15 objects. Second we use three of the features in the first set of objects along with different features from the second set. Therefore, we would have two sets of objects with different types of features for each set. In order to transfer knowledge from the source dataset to the target dataset, both the source dataset and the target dataset need to be related to each other. For example, objects that have different types and have different features need to share similar behaviors when exposed to the same environmental effects in order to transfer knowledge between them. In our experiment, the features of objects in the source dataset change in a similar way to the features of objects in the target dataset when exposed to the same environmental effects even though the features are of different types and have different range of values. To simulate the different amounts of relation between the source dataset and the target dataset, we introduce different degrees of noise to the fingerprints of the source objects starting from small amount of added noise referred to as Class 1 and all the way to large amount of added noise referred to as Class 5. Fig.~\ref{fig:tran} shows the effect of transfer learning as the relationship between the source and target datasets decreases from Class 1 to Class 5. The figure shows a small average distance between fingerprints when the relation between the source and target dataset is high, as in Classes 1 to 3 where we get around 70\% improvement over the no transfer approach. However, when the relationship between the source and target datasets is small, as in Classes 4 and 5 in Fig.~\ref{fig:tran}, the transfer learning method increases the average distance and yields a lower detection rate compared to having no transfer learning. This case is called negative transfer when forcing a transfer learning method even though the source and target dataset are not related to each other. The figure shows two values of the transfer weight $\alpha$, 25\% or 50\% in Fig.~\ref{fig:tran}, that represent the amount of weight the transfer learning method has on the final assignment of fingerprints. The figure shows that relying more on transfer learning is better when the relationship between the source and target datasets is high, as shown for Classes 1 and 2 in Fig.~\ref{fig:tran}. In contrast, relying on transfer learning has a negative impact when the relationship between the source and target datasets is low, as shown for Classes 4 and 5 in Fig.~\ref{fig:tran}. The figure shows an overall improvement of 80\% in the average distance when the relationship between the source and target datasets is high.

\begin{figure}[!t]
  \begin{center}
	\centering
 	   \includegraphics[width=8cm]{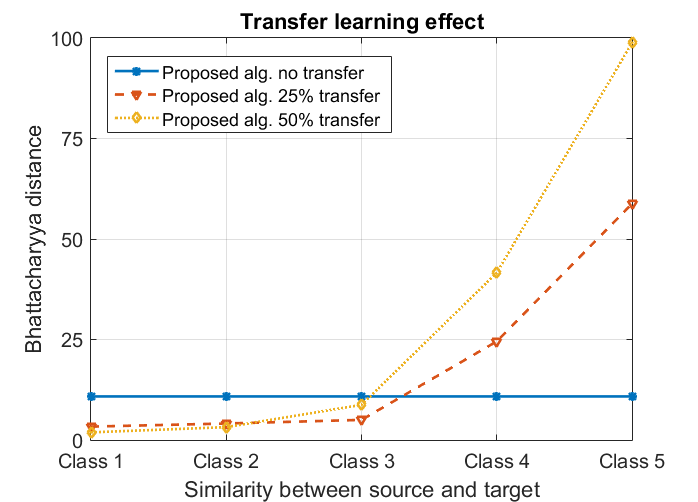}\vspace{-0.2cm}
  	  \caption{\small{\label{fig:tran}The effect of transfer learning between source and target datasets as the relationship between the source and target datasets changes from the highest relation in Class 1 to the lowest relation in Class 5.}}
  \end{center}
\end{figure}

\vspace{-0.5cm}
\section{Conclusion}\label{sec:end}
In this paper, we have proposed a novel authentication framework for IoT systems. The proposed framework exploits the effects of the environment surrounding IoT objects to detect remote emulation attackers, who can replicate the signals of legitimate objects but fail to replicate the constantly changing environment around the legitimate IoT object. The proposed framework tracks the changes in objects' fingerprints and uses these changes to extract a model for the environment. We have shown that our framework can enhance the authentication of legitimate objects and is able to detect both cyber and cyber-physical emulation attacks. The proposed framework used transfer learning as well to estimate the environment from objects of different types or objects with different feature spaces.  Simulation results using real IoT device data have shown that conventional methods of device fingerprinting were unable to detect cyber-physical emulation attackers while our proposed framework were able to detect these type of attacks. The results have shown an improvement of 40\% in cyber emulation attacks detection, and the transfer learning results showed an improvement of 70\% when the framework used different types of objects compared to using the framework only on objects with the exact same feature space.
\vspace{-0.2cm}

% or


%\appendices
%\section{}
%% you can choose not to have a title for an appendix
%% if you want by leaving the argument blank
%\section{}
%% use section* for acknowledgment
%\section*{Acknowledgment}

\ifCLASSOPTIONcaptionsoff
  \newpage
\fi



%\bibliographystyle{IEEEtran}

\bibliographystyle{IEEEtran}

\bibliography{IEEEfull}

%\begin{thebibliography}{1}
%\end{thebibliography}


%\end{IEEEbiography}
%% insert where needed to balance the two columns on the last page with
%%\newpage


%\vfill


\end{document}